\documentclass[journal]{IEEEtran}
\usepackage{graphicx}
\usepackage{amsfonts}
\usepackage{amsmath}
\usepackage{amssymb}
\usepackage{color}
\usepackage{bm}
\usepackage{bbm}
\usepackage{mathrsfs}
\usepackage{array}
\usepackage{booktabs}
\usepackage{multirow}
\newtheorem{thm}{Theorem}
\newtheorem{cor}{Corollary}
\newtheorem{lem}{Lemma}
\newtheorem{proof}{proof}

\newtheorem{rem}{Remark}
\newtheorem{exam}{Example}

\renewcommand{\baselinestretch}{1.2}
\onecolumn

\begin{document}

\title{Systematic Single-Deletion Multiple-Substitution
Correcting Codes}

\author{Wentu~Song, Nikita~Polyanskii, Kui~Cai,
        ~\IEEEmembership{senior member,~IEEE},
        and Xuan~He
\thanks{Wentu~Song, Kui~Cai and Xuan He are with Singapore
        University of Technology and Design, Singapore, e-mail:
        wentu$\_$song@sutd.edu.sg, cai$\_$kui@sutd.edu.sg,
        helaoxuan@126.com;}
\thanks{Nikita~Polyanskii is with
        Technical University of Munich, Germany, and Skolkovo
        Institute of Science and Technology, Russia,
        e-mail: ~nikitapolyansky@gmail.com.}
\thanks{\emph{Corresponding author: Kui Cai}}
        }

\maketitle

\begin{abstract}
Recent work by Smagloy \emph{et al}. $($ISIT 2020$)$ shows that
the redundancy of a single-deletion $s$-substitution correcting
code is asymptotically at least $(s+1)\log n+o(\log n)$, where $n$
is the length of the codes. They also provide a construction of
single-deletion and single-substitution codes with redundancy
$6\log n+8$. In this paper, we propose a family of systematic
single-deletion $s$-substitution correcting codes of length $n$
with asymptotical redundancy at most $(3s+4)\log n+o(\log n)$ and
polynomial encoding/decoding complexity, where $s\geq 2$ is a
constant. Specifically, the encoding and decoding complexity of
the proposed codes are $O\left(n^{s+3}\right)$ and
$O\left(n^{s+2}\right)$, respectively.
\end{abstract}


\IEEEpeerreviewmaketitle

\section{Introduction}

The problem of constructing deletion/insertion correcting codes
was introduced by Levenshtein \cite{Levenshtein65} and recently
has attracted an increasing attention due to their relevance to
the DNA-based data storage \cite{Heckel20}. In
\cite{Levenshtein65}, Levenshtein proved that a code can correct
up to $t$ deletions if and only if it can correct up to $t$
insertions, if and only if it can correct the combination of $t_1$
insertions and $t_2$ deletions for any non-negative integers $t_1$
and $t_2$ such that $t_1+t_2\leq t$. Levenshtein also proved that
for any $t$-deletion correcting code $\mathscr C$ of length $n$,
the redundancy of $\mathscr C~($defined as $n-\log|\mathscr C|)$
is asymptotically at least $t\log n+o(\log n)$, and an optimal
$t$-deletion correcting code has redundancy at most $2t\log
n+o(\log n)$.

The first class of optimal single-deletion correcting codes, whose
redundancy is $\log n+O(1)$, are the well-known
Varshamov-Tenengolts (VT) codes \cite{Varshamov65}, which are
defined as
$$\text{VT}_a(n)=\left\{\textbf{c}\in\{0,1\}^n: \textbf{c}\cdot
\textbf{v}\equiv a ~\text{mod}~(n+1)\right\},$$ where
$\textbf{v}=(1,2,\ldots,n)$, $a\in[0,n]$ and $\cdot$ is the
standard inner product. A decoding algorithm of the VT codes to
correct a single deletion was proposed in \cite{Levenshtein65},
and a systematic encoding algorithm of the VT codes was proposed
in \cite{Abdel98}, both have linear-time complexity.

The VT construction was generalized in \cite{Helberg02} and
\cite{Ghaffar12} by replacing the weight vector
$\textbf{v}=(1,2,\ldots,n)$ in the parity check equation with a
$t$-order recursive sequence. The resulted codes are capable of
correcting $t$ deletions. However, the asymptotic rate of such
codes is bounded away from $1$.

Multiple-deletion correcting codes with small asymptotical
redundancy were studied in several recent works. In
\cite{Brakensiek18}, Brakensiek \emph{et al}. presented a family
of $t$-deletion correcting codes with asymptotical redundancy
$O(t^2\log t\log n)$. For $t=2$, the redundancy was improved by
the works of Gabrys \emph{et al}. \cite{Gabrys19} and Sima
\emph{et al}. \cite{Sima19-1}. Specifically, the code in
\cite{Gabrys19} has redundancy $8\log n+O(\log\log n)$ and the
code in \cite{Sima19-1} has redundancy $7\log n+o(\log n)$. In a
more recent work by Guruswami and H\aa stad \cite{Gur2020}, an
explicit construction of $2$-deletion correcting codes with
redundancy $4\log n+o(\log n)$ was proposed, which matches the
existential upper bound of the optimal codes. However, it is not
known that whether $4\log n$ corresponds to the redundancy of the
optimal code, and it could be $a\log n$, where $a=2,3$ or $4$. For
general $t$, Sima and Bruck generalized the construction in
\cite{Sima19-1} to $t$-deletion correcting codes with redundancy
$8t\log n+o(\log n)$ \cite{Sima19-2}, where $n$ is the number of
information bits, and by further generalizing and applying the
techniques in \cite{Sima19-2}, Sima \emph{et al}. provided a
family of systematic $t$-deletion correcting codes with $4t\log
n+o(\log n)$ bits of redundancy and $O(n^{2t+1})$
encoding$/$decoding complexity \cite{Sima20}.

However, in many application scenarios, such as DNA data storage
and file synchronization, it is necessary to correct the edit
errors (i.e., the combination of insertions, deletions and
substitutions), which motivates the problem of constructing codes
that can correct insertions, deletions and/or substitutions. A
modified VT construction with redundancy $\log n+O(1)$ was
presented in \cite{Levenshtein65} to correct a single insertion,
deletion or substitution, which is also referred to as an edit.
Quaternary codes that can correct a single edit for DNA data
storage were considered in \cite{Cai19}. In \cite{Smagloy20}, a
family of single-deletion single-substitution correcting codes
(i.e., codes that can correct a single deletion and a single
substitution simultaneously) with redundancy $6\log n+8$ was
constructed using four VT-like parity check equations. The codes
constructed in \cite{Sima20} are capable of correcting combination
of insertions, deletions and substitutions such that the total
number of insertions, deletions and substitutions is upper bounded
by $k$. Such codes have redundancy $4k\log n+o(\log n)$, which is
the best known construction with respect to redundancy.

In this paper, we study the problem of constructing
single-deletion $s$-substitution correcting codes, i.e., codes
that can correct the combination of a single deletion and up to
$s$ substitutions. It was shown by Smagloy \emph{et.al.} in
\cite{Smagloy20} that the redundancy $r$ of such codes satisfies
$$r\geq (s+1)\log n+o(\log n).$$
The main result of this paper is a construction of a family of
single-deletion $s$-substitution correcting codes with a
systematic encoding function and with redundancy $r$ satisfying
$$r\leq (3s+4)\log n+o(\log n).$$
The encoding and decoding complexity of the proposed
codes are $O\left(n^{s+3}\right)$ and $O\left(n^{s+2}\right)$,
respectively. On the other hand, systematic codes are desirable in
practice since the information sequence can be extracted directly
from a codeword during decoding.

Our construction uses a set of higher order weight vectors,
denoted by $\textbf{a}^{(j)}$, $j=0,1,\ldots,2s+1$, to construct
parity checks $($or redundancies$)$, where $\textbf{a}^{(0)}$ is
the all-ones vector of length $n$ and
$\textbf{a}^{(j)}=\left(1^{j-1}, 1^{j-1}+2^{j-1}, \ldots,
\sum_{i=1}^ni^{j-1}\right)$, $j=1,2,\ldots,2s+1$. Similar higher
order weight vectors are used in the construction of
single-deletion single-substitution correcting codes
\cite{Smagloy20} and the construction of $t$-deletion correcting
codes \cite{Sima19-2}, \cite{Sima20}. We prove that it is possible
to use less redundancies $($than that used in \cite{Smagloy20} and
\cite{Sima20}$)$ to construct single-deletion
multiple-substitution correcting codes. According to the
construction in \cite{Sima20}, there exist codes correcting a
single deletion and $s$ substitutions with redundancy $4(s+1)\log
n+ o(\log n)$. In this work, by using a pre-coding function of BCH
code, our construction achieves the redundancy of $(3s+4)\log n+
o(\log n)$, decreasing by $s\log n$. We remark that by similar
discussions as in \cite{Levenshtein65}, it can be proven that
there exists a single-deletion $s$-substitution correcting code
with redundancy $2(s+1)\log n+o(\log n)$, however, the
encoding/decoding complexity of such a code are exponential, i.e.,
$O(n^{2s+2}2^n)$, which is not applicable in practice.

\subsection{Organization}
The $t$-deletion $s$-substitution correcting codes and relevant
concepts are introduced in Section
\uppercase\expandafter{\romannumeral 2}. Construction of
single-deletion $s$-substitution correcting codes is presented in
Section \uppercase\expandafter{\romannumeral 3}, and the related
lemmas used by our construction are proved in Section
\uppercase\expandafter{\romannumeral 4}. Finally, the paper is
concluded in Section \uppercase\expandafter{\romannumeral 5}.

\subsection{Notations}
The following notations are used in this paper:

 1) For any integers $m$ and $n$ such that $m\leq n$, we denote
 $[m,n]=\{m,m+1,\ldots,n\}$ and call it an \emph{interval}.
 If $m>n$, let $[m,n]=\emptyset$. For simplicity, denote
 $[1,n]=[n]$ for any positive integer $n$.

 2) For any set $S$, we use $|S|$ to denote the size (cardinality)
 of $S$. We also use $|x|$ to denote the absolute value of any
 real number $x$. This will not cause confusion because we can
 easily find its meaning from the context.

 3) For any sets $S$ and $T$, $S\backslash T=\{s\in S: s\notin
 T\}$ is the set of elements of $S$ that do not belong to $T$.

 4) For any vector $\textbf{x}\in \mathbb A^n$ and $i\in[n]$,
 where $\mathbb A$ is a fixed alphabet, unless otherwise
 stated, $x_i$ is the $i$th coordinate of $\textbf{x}$. In other
 words, $\textbf{x}=(x_1,\ldots,x_n)$. Superscript is allowed in
 this notation. For example, if $\textbf{a}^{(j)}\in\mathbb A^n$,
 where $j$ is an integer, then we have
 $\textbf{a}^{(j)}=\big(a^{(j)}_1,a^{(j)}_2,\ldots,a^{(j)}_n\big)$.

 5) For any $\textbf{x}=(x_1, x_2, \ldots,
 x_L)\in\mathbb A^n$ and any $D=\{i_1, i_2, \ldots, i_d\}\subseteq
 [L]$ such that $i_1<i_2<\cdots<i_d$, we denote
 $\textbf{x}_D=(x_{i_1}, x_{i_2}, \ldots, x_{i_d})$.

\section{Deletion and Substitution Correcting Codes}

In this paper, for any positive integer $L$, a vector
$\textbf{x}\in\{0,1\}^L$ is also called a sequence (or a string).
A \emph{subsequence} of $\textbf{x}$ is any sequence obtained from
$\textbf{x}$ by deleting one or more symbols; a \emph{substring}
of $\textbf{x}$ is a subsequence made of some consecutive symbols
of $\textbf{x}$. In other words, $\textbf{y}$ is called a
subsequence of $\textbf{x}$ if $\textbf{y}=\textbf{x}_{D}$ for
some $D\subseteq[L]$; $\textbf{y}$ is called a substring of
$\textbf{x}$ if $\textbf{y}=\textbf{x}_{[i,j]}$ for some
$i,j\in[L]$ such that $i\leq j$. For convenience, if $\textbf{y}$
can be obtained from a subsequence of $\textbf{x}$ by $s$
substitutions, then we say that $\textbf{y}$ is a subsequence of
$\textbf{x}$ with $s$ substitutions.

Suppose $t$ and $s$ are two positive integers such that $t+s<L$.
For any $\textbf{x}\in\{0,1\}^L$, let $\mathscr
B_{t,s}(\textbf{x})$ be the set of all sequences that can be
obtained from $\textbf{x}$ by $t$ deletions (i.e., deletion of $t$
symbols) and at most $s$ substitutions (i.e., substitution of at
most $s$ symbols). Note that $\mathscr
B_{t,s}(\textbf{x})\subseteq\{0,1\}^{L-t}$. A code $\mathscr
C\subseteq\{0,1\}^L$ is called a $t$-\emph{deletion}
$s$-\emph{substitution correcting code} if for any
$\textbf{c}\in\mathscr C$, $\textbf{c}$ can be correctly recovered
from any $\textbf{y}\in\mathscr B_{t,s}(\textbf{c})$.
Equivalently, $\mathscr C\subseteq\{0,1\}^L$ is a $t$-deletion
$s$-substitution correcting code if and only if $\mathscr
B_{t,s}(\textbf{c})\cap\mathscr B_{t,s}(\textbf{c}')=\emptyset$
for any distinct $\textbf{c},\textbf{c}'\in\mathscr C$. It is also
easy to see that if $\mathscr B_{t,s}(\textbf{c})\cap\mathscr
B_{t,s}(\textbf{c}')=\emptyset$, then $\mathscr
B_{t',s}(\textbf{c})\cap\mathscr B_{t',s}(\textbf{c}')=\emptyset$.
Hence, if $\mathscr C$ is a $t$-deletion $s$-substitution
correcting code, then for any $\textbf{c}\in\mathscr C$,
$\textbf{c}$ can be correctly recovered from any erroneous copy
$\textbf{y}$ of $\textbf{c}$ with $t'$ deletions and $s'$
substitutions for any $t'\in[t]$ and $s'\in[s]$.

In this paper, we consider binary single-deletion $s$-substitution
correcting codes $($i.e., $t=1$ and $s\geq 1)$ and aim to
construct systematic single-deletion $s$-substitution correcting
codes with low redundancy. Hence, it is reasonable to assume that
$L>2s+1\geq 3$.

\begin{rem}\label{rem-Nx}
For any $\textbf{x},\textbf{x}'\in\{0,1\}^L$, if $\mathscr
B_{1,s}(\textbf{x})\cap\mathscr
B_{1,s}(\textbf{x}')\neq\emptyset$, then there exists a
$\textbf{y}\in\mathscr B_{1,s}(\textbf{x})\cap\mathscr
B_{1,s}(\textbf{x}')$, that is, $\textbf{y}$ can be obtained from
$\textbf{x}~($resp. $\textbf{x}')$ by one deletion and at most $s$
substitutions. In other words, $\textbf{x}'$ has a subsequence of
length $L-1$ that can be obtained from a subsequence of
$\textbf{x}$ of length $L-1$ by $2s'$ substitutions for some
$s'\in[s]$. Formally, there exist $i_{\text{del}},
i'_{\text{del}}\in[L]$ and a set $\{\lambda_1,\lambda_2,\ldots,
\lambda_{2s}\}\subseteq[L]\backslash\{i_{\text{del}}\}$, where
$\lambda_1<\lambda_2<\cdots<\lambda_{2s}$, such that
$\textbf{x}'_{[L]\backslash\{i'_{\text{del}}\}}$ can be obtained
from $\textbf{x}_{[L]\backslash\{i_{\text{del}}\}}$ by
substitution of at most $2s$ symbols in $\{x_{\lambda_1},
x_{\lambda_2}, \ldots, x_{\lambda_{2s}}\}$.
\end{rem}

Let $\mathscr C\subseteq\{0,1\}^L$ be a code of length $L$ and
$k\in[L]$. A set $I\subseteq[L]$ of size $|I|=k$ is said to be an
\emph{information set} of $\mathscr C$ if for every
$\textbf{u}\in\{0,1\}^k$, there is at least one codeword
$\textbf{c}\in\mathscr C$ such that $\textbf{c}_{I}=\textbf{u}$.
Clearly, $k\leq\log|\mathscr C|$.\footnote{In this paper, we
consider binary codes, so all logarithms are taken with base two.}
We are interested in the largest $k$ for which $\mathscr C$ has an
information set of size $k$. An encoding function $\mathcal
E:\{0,1\}^k\rightarrow\{0,1\}^L$ is said to be \emph{systematic}
on $I$ if for every $\textbf{u}\in\{0,1\}^k$,
$\textbf{c}_{I}=\textbf{u}$, where $\textbf{c}=\mathcal
E(\textbf{u})$. For a systematic code, the information sequence
$\textbf{u}$ can be extracted directly from a codeword when
decoding. Hence, systematic codes are desirable in practice.

For any code $\mathscr C\subseteq\{0,1\}^L$, the redundancy of
$\mathscr C$ is defined as $L-\log|\mathscr C|$. Clearly, if
$\mathscr C$ has an encoding function $\mathcal
E:\{0,1\}^k\rightarrow\{0,1\}^L$, then the redundancy of $\mathscr
C$ equals to $L-k$.

The following two lemmas will be used in the subsequent
discussions.
\begin{lem}\label{lem-Split}
Suppose $\textbf{x}=\big(\textbf{x}^{(1)}, \textbf{x}^{(2)},
\textbf{x}^{(3)}\big)\in\{0,1\}^N$, where each $\textbf{x}^{(i)}$
is a substring of $\textbf{x}$ of length greater than $s$,
$i=1,2,3$. Then each $\textbf{y}\in\mathscr B_{1,s}(\textbf{x})$
can be split into three substrings, say
$\textbf{y}=\big(\textbf{y}^{(1)}, \textbf{y}^{(2)},
\textbf{y}^{(3)}\big)$, such that $\textbf{y}^{(i)}\in\mathscr
B_{1,s}(\textbf{x}^{(i)})$.
\end{lem}
\begin{proof}
Suppose $\textbf{x}=(x_1,x_2,\cdots,x_N)$ and
$\textbf{x}^{(i)}=x_{[N_{i-1}+1,N_i]}$, $i=1,2,3$, where
$0=N_0<N_1<N_2<N_3=N$. For each
$\textbf{y}=(y_1,y_2,\cdots,y_N)\in\mathscr B_{1,s}(\textbf{x})$
and each $i\in\{1,2,3\}$, we can let
$\textbf{y}^{(i)}=\textbf{y}_{[N_{i-1}+1,N_{i}-1]}$. If no
deletion occurs, then clearly, $\textbf{y}^{(i)}\in\mathscr
B_{1,s}(\textbf{x}^{(i)})$, $i=1,2,3$. Now, suppose a symbol $x_j$
is deleted from $\textbf{x}^{(i)}$ for some $i\in\{1,2,3\}$, then
we can see that $\textbf{y}^{(i)}$ is still a subsequence of
$\textbf{x}^{(i)}$, possibly with at most $s$ substitutions.
Hence, we always have $\textbf{y}^{(i)}\in\mathscr
B_{1,s}\left(\textbf{x}^{(i)}\right)$, $i=1,2,3$.
\end{proof}

\begin{lem}\label{lem-BCH}
For any positive integer $k$, there exists a positive integer
$n_0$ and a function $h: \{0,1\}^{k}\rightarrow\{0,1\}^{n_0}$ such
that $n_0-k\leq s(\log (n_0)+2)$ and $\mathscr
C_0\triangleq\{h(\textbf{x}): \textbf{x}\in\{0,1\}^{k}\}$ is a
systematic linear code of minimum $($Hamming$)$ distance at least
$2s+1$.
\end{lem}
\begin{proof}
Let $m$ be the smallest positive integer such that there exists a
binary narrow-sense primitive BCH code of length $2^m-1$, with
designed distance $\delta=2s+1$ and dimension $k'\geq k$. Let
$\mathscr C_{\text{BCH}}$ denote such a BCH code and
$h_{\text{BCH}}: \{0,1\}^{k'}\rightarrow\{0,1\}^{2^m-1}$ be a
systematic encoding function of $\mathscr C_{\text{BCH}}$. Then
$k'\geq 2^m-1-sm$ and the minimum $($Hamming$)$ distance of
$\mathscr C_{\text{BCH}}$ is at least $\delta=2s+1~($e.g., see
\cite{MacWilliams}$)$. For each $\textbf{x}\in\{0,1\}^{k}$, we can
obtain a codeword $(\textbf{0},\textbf{x},\textbf{p})$ of
$\mathscr C_{\text{BCH}}$ such that
$(\textbf{0},\textbf{x},\textbf{p})=h_{\text{BCH}}(\textbf{0},
\textbf{x})$, where $\textbf{0}$ is an all-zero vector of length
$k'-k$. Let $h(\textbf{x})=(\textbf{x},\textbf{p})$. Then we
obtain a function $h:\{0,1\}^{k}\rightarrow\{0,1\}^{n_0}$, where
$n_0=2^m-1-(k'-k)$. Clearly, $\{h(\textbf{x}):
\textbf{x}\in\{0,1\}^{k}\}$ is a systematic linear code of minimum
$($Hamming$)$ distance at least $2s+1$.

Since $k'\geq 2^m-1-sm$, we have
$n_0-k=2^m-1-(k'-k)-k=2^m-1-k'\leq sm$. Moreover, note that $k\leq
k'<2^m-1$ and $2^{m-1}-1-s(m-1)<k$ $($Otherwise, the BCH code of
length $2^{m-1}-1$ and designed distance $\delta=2s+1$ can be used
\cite{MacWilliams}, which contradicts to the minimality of
$m$.$)$. We have $k'-k\leq
2^m-1-(2^{m-1}-1-s(m-1))=2^{m-1}+s(m-1)$. Then
$n_0=2^m-1-(k'-k)\geq 2^{m-1}-1-s(m-1)\geq 2^{m-2}$, and so $\log
(n_0)\geq m-2$. Hence, $n_0-k\leq sm\leq s(\log (n_0)+2)$.
\end{proof}

\section{Construction of Single-Deletion $s$-Substitution
Correcting Codes}

In this section, we propose a family of systematic single-deletion
$s$-substitution correcting codes. The length of the information
sequences is denoted by $k$, where $k$ is a positive integer. Let
$h: \{0,1\}^{k}\rightarrow\{0,1\}^{n_0}$ be the function
constructed in Lemma \ref{lem-BCH}. Then the encoding function of
the proposed code, denoted by $\mathcal E$, is defined as
\begin{align}\label{en-fun}
\mathcal E(\textbf{x})=\bigg(h(\textbf{x}),g(h(\textbf{x})),
\text{Rep}_{2s+2}\Big(f\big(g(h(\textbf{x}))\big)\Big)\bigg),~~
\forall~\textbf{x}\in\{0,1\}^{k},
\end{align}
where $\text{Rep}_{2s+2}(\cdot)$ is the encoding function of the
$(2s+2)$-fold repetition code. The functions $g:
\{0,1\}^{n_0}\rightarrow\{0,1\}^{n_1}$ and $f:
\{0,1\}^{n_1}\rightarrow\{0,1\}^{n_2}$, which will be constructed
later, satisfy the following three conditions:\\
 $($C1$)$ $n_1=2(s+2)\log (n_0)+o(\log (n_0))$ and $n_2=o(\log (n_0))$;\\
 $($C2$)$ For every $\textbf{x}\in\{0,1\}^{k}$,
 $h(\textbf{x})$ can be recovered from $g\big(h(\textbf{x})\big)$ and
 any given sequence in $\mathscr
 B_{1,s}\big(h(\textbf{x})\big)$;\\
 $($C3$)$ For every $\textbf{x}\in\{0,1\}^{k}$,
 $g\big(h(\textbf{x})\big)$ can be recovered from
 $f\Big(g(h\big(\textbf{x})\big)\Big)$ and any given sequence in
 $\mathscr B_{1,s}\Big(g(h\big(\textbf{x})\big)\Big)$.

{\vskip 2pt}From any $\textbf{y}\in\mathscr B_{1,s}\big(\mathcal
E(\textbf{x})\big)$, $\textbf{x}$ can be recovered as follows: By
Lemma \ref{lem-Split}, we can obtain three sequences
$\textbf{y}^{(1)},\textbf{y}^{(2)},\textbf{y}^{(3)}$ such that
$\textbf{y}^{(1)}\in\mathscr B_{1,s}\big(h(\textbf{x})\big)$,
$\textbf{y}^{(2)}\in\mathscr
B_{1,s}\Big(g\big(h(\textbf{x})\big)\Big)$ and
$\textbf{y}^{(3)}\in\mathscr
B_{1,s}\bigg(\text{Rep}_{2s+2}\Big(f\big(g(h(\textbf{x})
)\big)\Big)\bigg)$. Note that
$f\Big(g\big(h(\textbf{x})\big)\Big)$ can always be recovered from
$\textbf{y}^{(3)}$. Then by condition $($C3$)$,
$g\big(h(\textbf{x})\big)$ can be recovered from
$f\Big(g(h\big(\textbf{x})\big)\Big)$ and $\textbf{y}^{(2)}$.
Further, by condition $($C2$)$, $h(\textbf{x})$ can be recovered
from $g\big(h(\textbf{x})\big)$ and $\textbf{y}^{(1)}$. Finally,
$\textbf{x}$ can be recovered from $h(\textbf{x})$ by the inverse
function $h^{-1}$ of $h$. Hence, the encoding function $\mathcal
E$ gives a single-deletion $s$-substitution correcting code. Since
by Lemma \ref{lem-BCH}, $h$ is a systematic encoding function, so
$\mathcal E$ is also a systematic encoding function. Moreover, by
the construction, the proposed code has length $n=n_0+n_1+n_2\geq
n_0$, so by Lemma \ref{lem-BCH} and condition $($C1$)$, its
redundancy $r$ satisfies $r\leq n_0-k+n_1+n_2\leq s(\log
(n_0)+2)+2(s+2)\log (n_0)+o(\log (n_0))=(3s+4)\log (n_0)+o(\log
(n_0))\leq (3s+4)\log n+o(\log n)$.

In the rest of this section, we will provide the construction of
$f$ and $g$ by two lemmas $($i.e., Lemma \ref{f-protect} and Lemma
\ref{M-cmprsn}$)$, and then present our main result $($i.e.,
Theorem \ref{main-thm}$)$. We will prove Lemma \ref{f-protect} and
Lemma \ref{M-cmprsn} in Section
\uppercase\expandafter{\romannumeral 4}.

For generality, let $L\geq 3$ be an arbitrarily fixed integer.
Then we will construct a function
$f:\{0,1\}^L\rightarrow\{0,1\}^{\lceil\xi(L)\rceil}$, where
$\xi(L)=(s+1)(2s+1)\log L+(2s+1)\log (2s+1)$ and
$\lceil\cdot\rceil$ is the ceiling function. To construct $f$, we
first define a set of vectors
$\left\{\textbf{a}^{(j)}=\left(a^{(j)}_1, a^{(j)}_2, \ldots,
a^{(j)}_L\right): j\in[2s+1]\right\}$ such that for each
$j\in[2s+1]$,
$$\textbf{a}^{(j)}=\left(1^{j-1}, \sum_{\ell=1}^2\ell^{j-1}, \ldots,
\sum_{\ell=1}^L\ell^{j-1}\right).$$ That is,
$a^{(j)}_i=\sum_{\ell=1}^i \ell^{j-1}$ for each $j\in[2s+1]$ and
$i\in[L]$. Then for each $\textbf{x}\in\{0,1\}^L$, define
$f(\textbf{x})=\big(f(\textbf{x})_1,f(\textbf{x})_2,\ldots,
f(\textbf{x})_{2s+1}\big)$ such that for each $j\in[2s+1]$,
\begin{align}\label{def-f}
f(\textbf{x})_j=\textbf{x}\cdot\textbf{a}^{(j)}~~
\text{mod}~(2s+1)L^j,\end{align} where
$\textbf{x}\cdot\textbf{a}^{(j)}$ is the standard inner product of
$\textbf{x}$ and $\textbf{a}^{(j)}$.

By the construction, we have $f(\textbf{x}) \in\mathscr M$, where
$$\mathscr
M\triangleq\prod_{j=1}^{2s+1}\left[0,(2s+1)L^j-1\right]$$ and
$\prod$ is the Cartesian product of sets. For each
${\textbf{r}}=(r_1,r_2,\ldots,r_{2s+1}) \in\mathscr M$, define
$$M(\textbf{r})=\sum_{j=1}^{2s+1}\left(r_j\prod_{i=0}^{j-1}
(2s+1)L^i\right).$$ Then we obtain a bijection
\begin{align}\label{def-M-map} M: \mathscr M\rightarrow
\left[0,(2s+1)^{2s+1}L^{(s+1)(2s+1)}-1\right],\end{align} and each
$f(\textbf{x})$ can be viewed as a binary sequence of length
$\lceil\xi(L)\rceil$, i.e., the binary representation of
$M(f(\textbf{x}))$, where $\xi(L) =(s+1)(2s+1)\log
L+(2s+1)\log(2s+1)$. In this paper, we can safely identify
$f(\textbf{x})$ and the binary representation of
$M(f(\textbf{x}))$. Hence, we obtain a function
$f:\{0,1\}^L\rightarrow\{0,1\}^{\lceil\xi(L)\rceil}$, where
$\xi(L) =(s+1)(2s+1)\log L+(2s+1)\log (2s+1)$.

The construction of $f$ is similar to the Sima-Bruck-Gabrys
construction in \cite{Sima20}. For the case of $t=1$, the
construction in \cite{Sima20} consists of $2(s+1)+1=2s+3$
components, that is,
$f(\textbf{x})=\big(f(\textbf{x})_1,f(\textbf{x})_2,\ldots,
f(\textbf{x})_{2s+3}\big)$, while in this paper, we prove that
$2s+1$ components are sufficient, that is, we only need
$f(\textbf{x})=\big(f(\textbf{x})_1,f(\textbf{x})_2,\ldots,
f(\textbf{x})_{2s+1}\big)$. In fact, the following lemma shows
that $\textbf{x}$ can be protected by $f(\textbf{x})$ from a
single deletion and $s$ substitutions.

\begin{lem}\label{f-protect}
For any $\textbf{x}$, $\textbf{x}'\in\{0,1\}^L$, if $\mathscr
B_{1,s}(\textbf{x})\cap\mathscr B_{1,s}(\textbf{x}')\neq\emptyset$
and $f(\textbf{x})=f(\textbf{x}')$, then $\textbf{x}=\textbf{x}'$.
\end{lem}

By Lemma \ref{f-protect}, for any $\textbf{x}\in\{0,1\}^L$,
$\textbf{x}$ is uniquely determined by $f(\textbf{x})$ and any
given $\textbf{y}\in\mathscr B_{1,s}(\textbf{x})$, so $\textbf{x}$
can be recovered from $f(\textbf{x})$ and $\textbf{y}$. In other
words, $\textbf{x}$ can be protected by $f(\textbf{x})$ from a
single deletion and $s$ substitutions.

\begin{rem}\label{rem-f-protect}
Using Lemma \ref{f-protect}, we can give a construction of
single-deletion $s$-substitution correcting code as follows. Take
$L=n$ as the length of the code to be constructed. For any fixed
$\textbf{r}=(r_1,r_2,\ldots,r_{2s+1}) \in\mathscr
M=\prod_{j=1}^{2s+1}\left[0,(2s+1)n^j-1\right]$, let $$\mathscr
C_{\textbf{r}}=\big\{\textbf{c}\in\{0,1\}^n:
f(\textbf{c})=\textbf{r}\big\}.$$ By Lemma \ref{f-protect},
$\mathscr B_{1,s}(\textbf{c})\cap\mathscr
B_{1,s}(\textbf{c}')=\emptyset$ for any distinct
$\textbf{c},\textbf{c}'\in\mathscr C_{\textbf{r}}$. Hence,
$\mathscr C_{\textbf{r}}$ is a single-deletion $s$-substitution
correcting code. Since the number of $\textbf{r}\in\mathscr M$ is
\begin{align*}
|\mathscr
M|=\left|\prod_{j=1}^{2s+1}\left[0,(2s+1)n^j-1\right]\right|=
\prod_{j=1}^{2s+1}(2s+1)n^j, 
\end{align*} then by the
pigeonhole principle, there exists an $\textbf{r}$ such that
$|\mathscr C_{\textbf{r}}|\geq\frac{2^n}{|\mathscr M|}=
\frac{2^n}{\prod_{j=1}^{2s+1}(2s+1)n^j}$. Hence, the redundancy
$r(\mathscr C_{\textbf{r}})$ of $\mathscr C_{\textbf{r}}$
satisfies \begin{align*}r(\mathscr
C_{\textbf{r}})&\leq\log\left(\prod_{j=1}^{2s+1}(2s+1)n^j\right)
\\&=\sum_{j=1}^{2s+1}j\log
n+(2s+1)\log(2s+1)\\&=(s+1)(2s+1)\log
n+(2s+1)\log(2s+1).\end{align*}
\end{rem}

For $s=1$, $\mathscr C_{\textbf{r}}$ is a single-deletion
single-substitution correcting code with redundancy $r(\mathscr
C_{\textbf{r}})\leq 6\log n+3$, which is stated as the following
corollary.

\begin{cor}\label{code-cor}
For any positive integer $n$ such that $n>6\log n+3$, there exists
a single-deletion single-substitution correcting code of length
$n$ and at most $6\log n+3$ redundancy bits.
\end{cor}

A similar construction is Construction 11 of \cite{Smagloy20}, in
which four components $f(\textbf{x})_j, j=0,1,2,3,$ are used and
the redundancy of the corresponding code is at most $6\log n+8$,
where $f(\textbf{x})_0=\sum_{i=1}^nx_i~\text{mod}~5$. However, the
proof of Lemma \ref{f-protect} shows that $f(\textbf{x})_0$ is in
fact not necessary and $f(\textbf{x})_j, j=1,2,3,$ are sufficient.

Note that for $s=1$, the redundancy of $\mathscr C_{\textbf{r}}$
is smaller than the redundancy of the construction in
\cite{Sima20}, which is $4(s+1)+o(\log n)=8\log n+o(\log n)$.
Unfortunately, for $s\geq 2$, $r(\mathscr C_{\textbf{r}})$ is
larger than the redundancy of the construction in \cite{Sima20}.
In the following, we will always assume that $s\geq 2$. Lemma
\ref{M-cmprsn} gives the construction of the function $g$ using
the syndrome compression technique, which was first introduced by
Sima \emph{et al}. \cite{Sima19-2}. We will show that the
pre-coding function $h$ makes the redundancy of the resulted
single-deletion $s$-substitution correcting code $s\log n$ smaller
than the construction in \cite{Sima20}.

\begin{lem}\label{M-cmprsn}
Let $h$ be the function constructed as in Lemma \ref{lem-BCH}.
There exists a function
$$g:\{0,1\}^{n_0}\rightarrow \{0,1\}^{n_1},$$ where $n_1=2(s+2)\log
(n_0)+o(\log (n_0))$, such that for any
$\textbf{x}\in\{0,1\}^{k}$, $h(\textbf{x})$ can be recovered from
$g(h(\textbf{x}))$ and any $\textbf{y}\in \mathscr
B_{1,s}(h(\textbf{x}))$. Moreover, $g(h(\textbf{x}))$ can be
computed in time $O\left((n_0)^{s+3}\right)$, and $h(\textbf{x})$
can be computed from $g(h(\textbf{x}))$ and $\textbf{y}$ in time
$O\left((n_0)^{s+2}\right)$.
\end{lem}

Now, we can present our main result of this paper.
\begin{thm}\label{main-thm}
Let $g:\{0,1\}^{n_0}\rightarrow \{0,1\}^{n_1}$ be constructed as
in Lemma \ref{M-cmprsn} and $f$ be constructed by \eqref{def-f}
with $L=n_1$. The code $\mathscr C$ with its encoding function
$\mathcal E$ given by \eqref{en-fun} is a systematic
single-deletion $s$-substitution correcting code of length
$n=n_0+n_1+n_2$. The redundancy $r(\mathscr C)$ of $\mathscr C$
satisfies
$$r(\mathscr C)\leq(3s+4)\log n+o(\log n),$$ and the encoding and
decoding complexity of $\mathscr C$ are $O\left(n^{s+3}\right)$
and $O\left(n^{s+2}\right)$, respectively.
\end{thm}
\begin{proof}
Previously, we have seen that the encoding function $\mathcal E$
defined by \eqref{en-fun} is systematic, so $\mathscr C$ is a
systematic code. By Lemma \ref{f-protect}, the function $f$
satisfies condition $($C3$)$ and by Lemma \ref{M-cmprsn}, the
function $g$ satisfies conditions $($C1$)$ and $($C2$)$, so by our
previous discussions, $\mathscr C$ is a single-deletion
$s$-substitution correcting code. Clearly, by \eqref{en-fun}, the
length of $\mathscr C$ is $n=n_0+n_1+n_2$.

By Lemma \ref{lem-BCH}, the redundancy $r(\mathscr C_0)$ of
$\mathscr C_0\triangleq\{h(\textbf{x}):
\textbf{x}\in\{0,1\}^{k}\}$ satisfies $r(\mathscr C_0)=n_0-k\leq
s(\log (n_0)+2)$. By Lemma \ref{M-cmprsn}, $n_1=2(s+2)\log
(n_0)+o(\log (n_0))$. Moreover, since $L=n_1$, by the construction
of $f$, we have $f:\{0,1\}^{n_1}\rightarrow\{0,1\}^{n_2}$, where
$n_2 =(s+1)(2s+1)\log n_1+(2s+1)\log (2s+1)=o(\log (n_0))$, so the
length of $\text{Rep}_{2s+2}\Big(f\big(g(h(\textbf{x}))\big)\Big)$
is $\text{length of}~
\text{Rep}_{2s+2}\Big(h\big(g(h(\textbf{x}))\big)\Big)
=(2s+2)n_2=(2s+2)o(\log (n_0))=o(\log (n_0)).$ Hence, the
redundancy $r(\mathscr C)$ of $\mathscr C$ satisfies
\begin{align*}r(\mathscr C)&\leq s(\log (n_0)+2)+2(s+2)\log (n_0)+o(\log
(n_0))\\&=(3s+4)\log (n_0)+o(\log (n_0))\\&\leq(3s+4)\log n+o(\log
n.\end{align*}

The encoding complexity of $\mathscr C$ is
$O\left((n_0)^{s+3}\right)=O\left(n^{s+3}\right)$, which comes
from the complexity of computing $g(h(\textbf{x}))$. Similarly,
the decoding complexity of $\mathscr C$ is
$O\left((n_0)^{s+2}\right)=O\left(n^{s+2}\right)$, which is due to
the complexity of computing $h(\textbf{x})$ from
$g(h(\textbf{x}))$ and $\textbf{y}\in \mathscr
B_{1,s}(h(\textbf{x}))$.
\end{proof}

For $s=1$, the code $\mathscr C_{\textbf{r}}$ constructed in
Remark \ref{rem-f-protect} has a smaller redundancy than the code
constructed in Theorem \ref{main-thm}. However, for $s>1$, Theorem
\ref{main-thm} gives a better construction than Remark
\ref{rem-f-protect} with respect to redundancy.

\section{Proof of Lemmas}

In this section, we prove Lemma \ref{f-protect} and Lemma
\ref{M-cmprsn}.

\subsection{Proof of Lemma \ref{f-protect}}

In this subsection, we prove Lemma \ref{f-protect}. Suppose
$\textbf{x}=(x_1,x_2,\cdots,x_L)$ and
$\textbf{x}'=(x'_1,x'_2,\cdots,x'_L)$. For each $i\in[L]$, let
\begin{align}\label{def-u-i}
u_i\triangleq\sum_{\ell=i}^Lx_\ell-\sum_{\ell=i}^Lx'_\ell.\end{align}
Then to prove $\textbf{x}=\textbf{x}'$, it suffices to prove
$u_i=0$ for all $i\in[L]$. The following two claims will be used
in later discussions.

\emph{Claim 1}: Suppose $\alpha,\beta\in[L]$ such that
$\alpha\leq\beta$ and $x_{i}=x'_{i}$ for each
$i\in[\alpha,\beta-1]$. Then for each $i\in[\alpha,\beta]$,
$u_i=u_{\beta}$.
\begin{proof}[Proof of Claim 1]
By assumption and by \eqref{def-u-i}, we have
\begin{align*}
u_i&=\sum_{\ell=i}^Lx_\ell-\sum_{\ell=i}^Lx'_\ell\\
&=\left(\sum_{\ell=i}^{\beta-1} x_\ell+\sum_{\ell=\beta}^L
x_\ell\right)-\left(\sum_{\ell=i}^{\beta-1}x'_\ell
+\sum_{\ell=\beta}^Lx'_\ell\right)\\
&=\left(\sum_{\ell=i}^{\beta-1}
x_\ell-\sum_{\ell=i}^{\beta-1}x'_\ell\right)+\left(\sum_{\ell=\beta}^L
x_\ell-\sum_{\ell=\beta}^Lx'_\ell\right)\\
&=\sum_{\ell=\beta}^L x_\ell-\sum_{\ell=\beta}^Lx'_\ell\\
&=u_{\beta},
\end{align*}
which proves Claim 1.
\end{proof}

\emph{Claim 2}: Suppose $\alpha,\beta\in[L]$ such that
$\alpha\leq\beta$ and $x_{i}=x'_{i-1}$ for all $i\in[\alpha+1,
\beta-1]$. Then either $u_i\geq 0$ for all $i\in [\alpha,\beta]$
or $u_i\leq 0$ for all $i\in [\alpha,\beta]$.

\begin{proof}[Proof of Claim 2]
For each $i\in[\alpha,\beta-1]$, by \eqref{def-u-i} and by
assumption, we have
\begin{align*}
u_i&=\sum_{\ell=i}^Lx_\ell-\sum_{\ell=i}^Lx'_\ell\\
&=\left(x_i+\sum_{\ell=i+1}^{\beta-1} x_\ell+\sum_{\ell=\beta}^L
x_\ell\right)-\left(\sum_{\ell=i}^{\beta-2}x'_\ell+x'_{\beta-1}
+\sum_{\ell=\beta}^Lx'_\ell\right)\\
&=x_i-x'_{\beta-1}+\left(\sum_{\ell=i+1}^{\beta-1}
x_\ell-\sum_{\ell=i}^{\beta-2}x'_\ell\right)+\left(\sum_{\ell=\beta}^L
x_\ell-\sum_{\ell=\beta}^Lx'_\ell\right)\\
&=x_i-x'_{\beta-1}+\left(\sum_{\ell=\beta+1}^L
x_\ell-\sum_{\ell=\beta}^Lx'_\ell\right)\\
&=x_i-x'_{\beta-1}+u_{\beta}.
\end{align*}
Since $x'_{\beta-1}\in\{0,1\}$, we can consider the following two
cases.

Case 1: $x'_{\beta-1}=0$. Since $x_i\in\{0,1\}$ for each
$i\in[\alpha,\beta-1]$, then $x_i-x'_{\beta-1}\in\{0,1\}$. Note
that $u_{\beta}$ is an integer, so if $u_{\beta}\geq 0$, then
$u_i\geq 0$ for all $i\in [\alpha,\beta]$; if $u_{\beta}\leq -1$,
then $u_i\leq 0$ for all $i\in [\alpha,\beta]$.

Case 2: $x'_{\beta-1}=1$. Then $x_i-x'_{\beta-1}\in\{0,-1\}$. If
$u_{\beta}\geq 1$, then $u_i\geq 0$ for all $i\in [\alpha,\beta]$;
if $u_{\beta}\leq 0$, then $u_i\leq 0$ for all $i\in
[\alpha,\beta]$.

Thus, it always holds that either $u_i\geq 0$ for all $i\in
[\alpha,\beta]$ or $u_i\leq 0$ for all $i\in [\alpha,\beta]$,
which proves Claim 2.
\end{proof}

It can be easily verified that as a special case of Claim 2, for
any $\beta\in[2,L]$ and $\alpha=\beta-1$, either $u_i\geq 0$ for
all $i\in [\alpha,\beta]$ or $u_i\leq 0$ for all $i\in
[\alpha,\beta]$. In fact, in this special case, we have
$[\alpha,\beta]=\{\beta-1,\beta\}$. Note that by \eqref{def-u-i},
$u_{\beta-1}=x_{\beta-1}-x'_{\beta-1}+u_\beta$, and note that
$x_{\beta-1}-x'_{\beta-1}\in\{-1,0,1\}$. Then both $u_{\beta-1}$
and $u_{\beta}$ are non-negative or $u_{\beta-1}$ and $u_{\beta}$
are non-positive.

Since $\mathscr B_{1,s}(\textbf{x})\cap\mathscr
B_{1,s}(\textbf{x}')\neq\emptyset$, then by Remark \ref{rem-Nx},
there exist $i_{\text{del}}, i'_{\text{del}}\in[L]$ and a set
$\{\lambda_1,\lambda_2,\ldots,
\lambda_{2s}\}\subseteq[L]\backslash\{i_{\text{del}}\}$, where
$\lambda_1<\lambda_2<\cdots<\lambda_{2s}$, such that
$\textbf{x}'_{[L]\backslash\{i'_{\text{del}}\}}$ can be obtained
from $\textbf{x}_{[L]\backslash\{i_{\text{del}}\}}$ by
substituting at most $2s$ symbols in $\{x_{\lambda_1},
x_{\lambda_2}, \ldots, x_{\lambda_{2s}}\}$. By symmetry, we can
assume, without loss of generality, that $i_{\text{del}}\leq
i'_{\text{del}}$. Considering the symbols without substitution, we
have
\begin{equation}\label{E-Eq-xxp}
x_i=\left\{\begin{aligned} &x'_i, ~~~~~\text{for}~i\in
\left([1,i_{\text{del}}-1]\cup[i'_{\text{del}}+1,L]\right)
\backslash \{\lambda_1,\lambda_2,\ldots,
\lambda_{2s}\},\\
&x'_{i-1}, ~~\text{for}~i\in
[i_{\text{del}}+1,i'_{\text{del}}]\backslash
\{\lambda_1,\lambda_2,\ldots, \lambda_{2s}\}.
\end{aligned}\right.
\end{equation}

\begin{exam}\label{exm-main-1}
Suppose $\textbf{x},\textbf{x}'\in\{0,1\}^{24}$ such that
$\textbf{x}'_{[24]\backslash\{18\}}$ can be obtained from
$\textbf{x}_{[24]\backslash\{8\}}$ by substitution of at most $6$
symbols in $\{x_{2},x_5,x_{11},x_{13},x_{16},x_{21}\}$. In this
example, $L=24$, $s=3$, $i_{\text{del}}=8$, $i'_{\text{del}}=18$
and $\{\lambda_1,\lambda_2,\ldots,
\lambda_{2s}\}=\{2,5,11,13,16,21\}$. We have $x_i=x'_i$ for
$i\in\{1,3,4,6,7,19,20,22,23,24\}$ and $x_i=x'_{i-1}$ for
$i\in\{9,10,12,14,15,17,18\}$. See Fig. \ref{fig1-xx} for an
illustration.
\end{exam}

\begin{figure}
\begin{center}
\includegraphics[width=14.8cm]{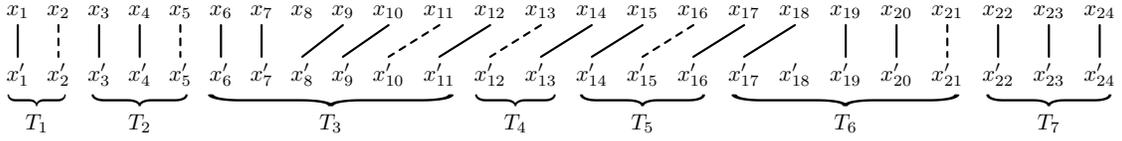}
\end{center}
\caption{Suppose $\textbf{x},\textbf{x}'\in\{0,1\}^{24}$ such that
$\textbf{x}'_{[24]\backslash\{8\}}$ can be obtained from
$\textbf{x}_{[24]\backslash\{18\}}$ by substituting at most $6$
symbols in $\{x_{2},x_5,x_{11},x_{13},x_{16},x_{21}\}$. In this
figure, each pair of the corresponding symbols of
$\textbf{x}_{[24]\backslash\{8\}}$ and
$\textbf{x}'_{[24]\backslash\{18\}}$ are connected by a (solid or
dashed) segment, where solid segments are for symbols without
substitution and dashed segments are for symbols possibly with
substitution. } \label{fig1-xx}
\end{figure}

Denote $\lambda_0=0$ and $\lambda_{2s+1}=L$. Let
$T_1,T_2,\cdots,T_{2s+1}$ be $2s+1$ subsets of $[L]$ defined as
\begin{align}\label{def-T}
T_e=[\lambda_{e-1}+1,\lambda_e], ~~\forall ~e\in[2s+1].\end{align}
It is easy to see that $T_1, T_2, \ldots, T_{2s+1}$ satisfy the
following simple properties:\\
 (P1) $T_e\neq\emptyset$ for $e\in[2s]$.\\
 (P2) $\{T_1, T_2, \ldots, T_{2s+1}\}$ is a partition of
 $[L]$, that is, $T_1, T_2, \ldots, T_{2s+1}$ are mutually
 disjoint and $\bigcup_{e=1}^{2s+1}T_e=[L]$.\\
 (P3) If $1\leq e<e'\leq 2s$, then $i_{e}<i_{e'}$ for any
 $i_e\in T_e$ and any $i_{e'}\in T_{e'}$.

We use Example \ref{exm-main-1} to show how the sets
$T_1,T_2,\cdots,T_{2s+1}$ are constructed. We have seen that
$\{\lambda_1,\lambda_2,\ldots,
\lambda_{2s}\}=\{2,5,11,13,16,21\}$. By \eqref{def-T}, we can
obtain $T_1=\{1,2\}$, $T_2=\{3,4,5\}$, $T_3=\{6,7,\cdots,11\}$,
$T_4=\{12,13\}$, $T_5=\{14,15,16\}$, $T_6=\{17,18,\cdots,21\}$ and
$T_7=\{22,23,24\}$. See Fig. \ref{fig1-xx} for an illustration.
Moreover, we have the following observations.
\begin{itemize}
 \item[1)] Consider $T_1=\{1,2\}$. By Fig. \ref{fig1-xx}, $x_1=x'_1$.
 Then by Claim 1, $u_1=u_2$. Hence, either $u_i\geq 0$ for all
 $i\in T_1$, or $u_i\leq 0$ for all $i\in T_1$.
 The same property holds for $T_2$ and $T_7$.
 \item[2)] Consider $T_3=\{6,7,\ldots,11\}$. First,
 by Fig. \ref{fig1-xx}, $x_{i}=x'_{i-1}$ for each
 $i\in\{9,10\}$, so by Claim 2, either $u_i\geq 0$
 for all $i\in \{8,9,10,11\}$, or $u_i\leq 0$ for all
 $i\in \{8,9,10,11\}$.
 Moreover, by Fig. \ref{fig1-xx}, $x_i=x'_i$ for each
 $i\in\{6,7\}$, so by Claim 1, $u_i=u_8$ for each
 $i\in \{6,7\}$. Hence, we have either $u_i\geq 0$ for all
 $i\in T_3$, or $u_i\leq 0$ for all $i\in T_3$.
 \item[3)] For $T_4=\{12,13\}$, by Claim 2, either $u_i\geq 0$
 for all $i\in T_4$, or $u_i\leq 0$ for all $i\in T_4$.
 This property also holds for $T_5$.
 \item[4)] Consider $T_6=\{17,18,\cdots,21\}$. First,
 by Claim 2, we have either $u_i\geq 0$ for all $i\in \{17,18,19\}$,
 or $u_i\leq 0$ for all $i\in \{17,18,19\}$. Moreover, by Claim 1,
 $u_i=u_{21}$ for each $i\in\{19,20,21\}$, or equivalently,
 $u_i=u_{19}$ for each $i\in\{19,20,21\}$. Hence,
 we have either $u_i\geq 0$ for all
 $i\in T_6$, or $u_i\leq 0$ for all $i\in T_6$.
\end{itemize}

Note that in Example \ref{exm-main-1}, $i_{\text{del}}=8$ and
$i'_{\text{del}}=18$, so we have $[i_{\text{del}}+1,
i'_{\text{del}}]\cap \{2,5,11,13,16,21\}=\{\lambda_{3},
\lambda_{4}, \lambda_{5}\}=\{11,13,16\}\neq\emptyset$. The
following is another example showing the construction of
$T_1,T_2,\cdots,T_{2s+1}$, where $[i_{\text{del}}+1,
i'_{\text{del}}]\cap\{\lambda_1,\lambda_2,\ldots,
\lambda_{2s}\}=\emptyset$.

\begin{exam}\label{exm-main-2}
Suppose $\textbf{x},\textbf{x}'\in\{0,1\}^{24}$ such that
$\textbf{x}'_{[24]\backslash\{15\}}$ can be obtained from
$\textbf{x}_{[24]\backslash\{11\}}$ by substituting at most $6$
symbols in $\{x_{2},x_5,x_{8},x_{17},x_{20},x_{22}\}$. In this
example, $L=24$, $s=3$, $\{\lambda_1,\lambda_2,\ldots,
\lambda_{2s}\}=\{2,5,8,17,20,22\}$, $i_{\text{del}}=11$ and
$i'_{\text{del}}=15$. By \eqref{def-T}, we can obtain
$T_1=\{1,2\}$, $T_2=\{3,4,5\}$, $T_3=\{6,7,8\}$,
$T_4=\{9,10,\cdots,17\}$, $T_5=\{18,19,20\}$, $T_6=\{21,22\}$ and
$T_7=\{23,24\}$. See Fig. \ref{fig2-xx} for an illustration. Note
that in this example, we have $[i_{\text{del}}+1,
i'_{\text{del}}]\cap \{\lambda_1,\lambda_2,\ldots,
\lambda_{2s}\}=\emptyset$. For each fixed $e\in\{1,2,3,5,6,7\}$,
by Claim 1, we can easily verify that either $u_i\geq 0$ for all
$i\in T_e$, or $u_i\leq 0$ for all $i\in T_e$. A new case of this
example is for $T_4$, which we state as the following observation.
\begin{itemize}
 \item[5)] Consider $T_4=\{9,10,\cdots,17\}$. First, by Claim 2,
 we can verify that
 either $u_i\geq 0$ for all $i\in \{11,12,\cdots,16\}$,
 or $u_i\leq 0$ for all $i\in \{11,12,\cdots,16\}$. Moreover,
 by Claim 1, we have
 $u_i=u_{11}$ for each $i\in\{9,10,11\}$ and $u_{17}=u_{16}$.
 Hence, we have either
 $u_i\geq 0$ for all $i\in T_4$, or $u_i\leq 0$ for all $i\in T_4$.
\end{itemize}
\end{exam}

\begin{figure}
\begin{center}
\includegraphics[width=14.8cm]{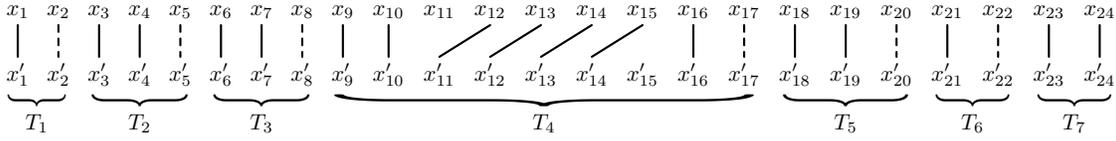}
\end{center}
\caption{Suppose $\textbf{x},\textbf{x}'\in\{0,1\}^{24}$ such that
$\textbf{x}'_{[24]\backslash\{15\}}$ can be obtained from
$\textbf{x}_{[24]\backslash\{11\}}$ by substituting at most $6$
symbols in $\{x_{2},x_5,x_{8},x_{17},x_{20},x_{22}\}$. Each pair
of the corresponding symbols of
$\textbf{x}_{[24]\backslash\{11\}}$ and
$\textbf{x}'_{[24]\backslash\{15\}}$ are connected by a (solid or
dashed) segment, where solid segments are for symbols without
substitution and dashed segments are for symbols possibly with
substitution. } \label{fig2-xx}
\end{figure}

In general, we have the following claim, which plays an important
role in our proof of Lemma \ref{f-protect}.

\emph{Claim 3}: For each fixed $e\in[1,2s+1]$, either $u_i\geq 0$
for all $i\in T_e$, or $u_i\leq 0$ for all $i\in T_e$.

\begin{proof}[Proof of Claim 3]
We need to consider the following two cases.

Case 1: $[i_{\text{del}}+1, i'_{\text{del}}]\cap S\neq\emptyset$.
Suppose $[i_{\text{del}}+1, i'_{\text{del}}]\cap
S=\{\lambda_{s_1},\ldots, \lambda_{s_2}\}$.\footnote{Note that
$s_1\leq s_2$. In fact, if $|[i_{\text{del}}+1,
i'_{\text{del}}]\cap S|=1$, then $s_1=s_2$.} Then we have
$\lambda_1<\cdots<\lambda_{s_1-1}<i_{\text{del}}<\lambda_{s_1}
<\cdots<\lambda_{s_2}\leq
i'_{\text{del}}<\lambda_{s_2+1}<\cdots<\lambda_{2s}$. We further
divide this case into the following four subcases.

Case 1.1: $e\in[1,s_1-1]\cup[s_2+2,2s+1]$. If $e\in[1,s_1-1]$,
then $\lambda_e<i_{\text{del}}$, and hence
$T_e=[\lambda_{e-1}+1,\lambda_e]\subseteq
[1,i_{\text{del}}-1]\backslash \{\lambda_1,\lambda_2,\ldots,
\lambda_{2s}\}$; if $e\in[s_2+2,2s+1]$, then
$\lambda_{e-1}>i'_{\text{del}}$, and hence
$T_e=[\lambda_{e-1}+1,\lambda_e]\subseteq[i'_{\text{del}}+1,L]
\backslash \{\lambda_1,\lambda_2,\ldots, \lambda_{2s}\}$. In both
cases, by \eqref{E-Eq-xxp}, $x_\ell=x'_\ell$ for each
$\ell\in[\lambda_{e-1}+1,\lambda_e-1]$, so by Claim 1,
$u_i=u_{\lambda_e}$ for each $i\in T_e$. Therefore, if
$u_{\lambda_e}\geq 0$, then $u_i\geq 0$ for all $i\in T_e$; if
$u_{\lambda_e}\leq 0$, then $u_i\leq 0$ for all $i\in T_e$.
$($This case is the same as observation 1) in Example
\ref{exm-main-1}.$)$

Case 1.2: $e=s_1$. By assumption,
$\lambda_{s_1-1}<i_{\text{del}}<\lambda_{s_1}\leq i'_{\text{del}}$
and $T_e=T_{s_1}=[\lambda_{s_1-1}+1,\lambda_{s_1}]
=[\lambda_{s_1-1}+1,i_{\text{del}}]\cup[i_{\text{del}},
\lambda_{s_1}]$.

First, consider $[i_{\text{del}}, \lambda_{s_1}]$. Note that
$[i_{\text{del}}+1,
\lambda_{s_1}-1]\subseteq[i_{\text{del}}+1,i'_{\text{del}}]
\backslash\{\lambda_1,\lambda_2,\ldots, \lambda_{2s}\}$. By
\eqref{E-Eq-xxp}, we have $x_\ell=x'_{\ell-1}$ for each
$\ell\in[i_{\text{del}}+1,\lambda_{s_1}-1]$, and so by Claim 2,
either $u_i\geq 0$ for all $i\in [i_{\text{del}},\lambda_{s_1}]$
or $u_i\leq 0$ for all $i\in [i_{\text{del}},\lambda_{s_1}]$.

Second, consider $[\lambda_{s_1-1}+1,i_{\text{del}}]$. Since
$[\lambda_{s_1-1}+1,i_{\text{del}}-1]\subseteq
[1,i_{\text{del}}-1]\backslash \{\lambda_1,\lambda_2,\ldots,
\lambda_{2s}\}~($recall that
$\lambda_{s_1-1}<i_{\text{del}}<\lambda_{s_1}\leq
i'_{\text{del}})$, by \eqref{E-Eq-xxp}, we have $x_\ell=x'_{\ell}$
for each $\ell\in[\lambda_{s_1-1}+1,i_{\text{del}}-1]$, and so by
Claim 1, we have $u_i=u_{i_{\text{del}}}$ for all $i\in
[\lambda_{s_1-1}+1,i_{\text{del}}]$.

Combining the above discussions, we proved that either $u_i\geq 0$
for all $i\in T_{s_1}$, or $u_i\leq 0$ for all $i\in T_{s_1}$.
$($This case is the same as observation 2) in Example
\ref{exm-main-1}.$)$

Case 1.3: $e\in[s_1+1,s_2]$. By assumption,
$i_{\text{del}}+1\leq\lambda_{e-1}<\lambda_e\leq i'_{\text{del}}$
and $T_e=[\lambda_{e-1}+1,\lambda_{e}]$. By \eqref{E-Eq-xxp}, we
have $x_\ell=x'_{\ell-1}$ for each
$\ell\in[\lambda_{e-1}+1,\lambda_{e}-1]$, so by Claim 2, either
$u_i\geq 0$ for all $i\in T_e$, or $u_i\leq 0$ for all $i\in T_e$.
$($This case is the same as observation 3) in Example
\ref{exm-main-1}.$)$

Case 1.4: $e=s_2+1$. By assumption, we have
$i_{\text{del}}<\lambda_{s_2}\leq
i'_{\text{del}}<\lambda_{s_2+1}$, and so
$T_e=T_{s_2+1}=[\lambda_{s_2}+1,\lambda_{s_2+1}]
=[\lambda_{s_2}+1,i'_{\text{del}}+1]\cup[i'_{\text{del}}+1,
\lambda_{s_2+1}]$.

First, consider $[\lambda_{s_2}+1,i'_{\text{del}}+1]$. By
\eqref{E-Eq-xxp}, we have $x_\ell=x'_{\ell-1}$ for each
$\ell\in[\lambda_{s_2}+1,i'_{\text{del}}]$, so by Claim 2, either
$u_i\geq 0$ for all $i\in[\lambda_{s_2}+1,i'_{\text{del}}+1]$, or
$u_i\leq 0$ for all $i\in[\lambda_{s_2}+1,i'_{\text{del}}+1]$.

Second, consider $[i'_{\text{del}}+1, \lambda_{s_2+1}]$. By
\eqref{E-Eq-xxp}, we have $x_\ell=x'_{\ell}$ for each
$\ell\in[i'_{\text{del}}+1, \lambda_{s_2+1}-1]$, so by Claim 1,
$u_i=u_{\lambda_{s_2+1}}$ for each $i\in[i'_{\text{del}}+1,
\lambda_{s_2+1}]$, or equivalently, $u_i=u_{i'_{\text{del}}+1}$
for each $i\in[i'_{\text{del}}+1, \lambda_{s_2+1}]$.

Combining the above discussions, we proved that either $u_i\geq 0$
for all $i\in T_{s_2+1}$ or $u_i\leq 0$ for all $i\in T_{s_2+1}$.
$($This case is the same as observation 4) in Example
\ref{exm-main-1}.$)$

Thus, for Case 1, we proved that for each fixed $e\in[2s+1]$,
either $u_i\geq 0$ for all $i\in T_{e}$ or $u_i\leq 0$ for all
$i\in T_{e}$.

Case 2: $[i_{\text{del}}+1, i'_{\text{del}}]\cap S=\emptyset$.
Suppose $\lambda_{s_0}<i_{\text{del}}\leq
i'_{\text{del}}<\lambda_{s_0+1}$, where $s_0\in[0,2s+1]$ and
$\lambda_{2s+2}=L+1$. We need to divide this case into the
following two subcases.

Case 2.1: $e\in[1,s_0]\cup[s_0+2,2s+1]$. By the same discussions
as in Case 1.1, we can prove that either $u_i\geq 0$ for all $i\in
T_{e}$ or $u_i\leq 0$ for all $i\in T_{e}$.

Case 2.2: $e=s_0+1$. In this case,
$T_e=T_{s_0+1}=[\lambda_{s_0}+1,\lambda_{s_0+1}]
=[\lambda_{s_0}+1,i_{\text{del}}]\cup[i_{\text{del}},
i'_{\text{del}}+1]\cup[i'_{\text{del}}+1,\lambda_{s_0+1}]$.

By the same discussions as in Case 1.4, we can prove that either
$u_i\geq 0$ for all $i\in [i_{\text{del}},
i'_{\text{del}}+1]\cup[i'_{\text{del}}+1,\lambda_{s_0+1}]$ or
$u_i\leq 0$ for all $i\in [i_{\text{del}},
i'_{\text{del}}+1]\cup[i'_{\text{del}}+1,\lambda_{s_0+1}]$.
Moreover, by \eqref{E-Eq-xxp}, we have $x_\ell=x'_{\ell}$ for each
$\ell\in[\lambda_{s_0}+1,i_{\text{del}}-1]$, so by Claim 1,
$u_i=u_{i_{\text{del}}}$ for each
$i\in[\lambda_{s_0}+1,i_{\text{del}}]$. Therefore, we have
$u_i\geq 0$ for all $i\in T_{s_0+1}$ or $u_i\leq 0$ for all $i\in
T_{s_0+1}$, where
$T_{s_0+1}=[\lambda_{s_0}+1,i_{\text{del}}]\cup[i_{\text{del}},
i'_{\text{del}}+1]\cup[i'_{\text{del}}+1,\lambda_{s_0+1}]$.
$($This case is the same as observation 5) in Example
\ref{exm-main-2}.$)$

Thus, for each fixed $e\in[2s+1]$, we have either $u_i\geq 0$ for
all $i\in T_{e}$ or $u_i\leq 0$ for all $i\in T_{e}$, which proves
Claim 3.
\end{proof}

Since by property (P2), the collection
$\{T_1,T_2,\ldots,T_{2s+1}\}$ is a partition of $[L]$, so to prove
$\textbf{x}=\textbf{x}'$, we only need to prove that $u_i=0$ for
each $e\in[2s+1]$ and each $i\in T_e$.

By the construction of $f$, for each $j\in[2s+1]$, we can compute
$\textbf{x}\cdot\textbf{a}^{(j)}-
\textbf{x}'\cdot\textbf{a}^{(j)}$ as follows.
\begin{align}\label{dff-xxp}
\textbf{x}\cdot\textbf{a}^{(j)}-
\textbf{x}'\cdot\textbf{a}^{(j)}&=\sum_{\ell=1}^Lx_\ell
a^{(j)}_\ell-\sum_{\ell=1}^Lx_\ell
a^{(j)}_\ell\nonumber\\&=\sum_{\ell=1}^Lx_\ell\left(\sum_{i=1}^\ell
i^{j-1}\right)
-\sum_{\ell=1}^Lx'_\ell\left(\sum_{i=1}^\ell i^{j-1}\right)\nonumber\\
&\stackrel{(\text{i})}{=}
\sum_{i=1}^L\left(\sum_{\ell=i}^Lx_\ell\right)i^{j-1}
-\sum_{i=1}^L\left(\sum_{\ell=i}^Lx'_\ell\right)i^{j-1}\nonumber\\
&=\sum_{i=1}^L\left(\sum_{\ell=i}^Lx_\ell
-\sum_{\ell=i}^Lx'_\ell\right)i^{j-1}\nonumber\\
&=\sum_{i=1}^Lu_ii^{j-1},
\end{align} where (i) comes from rearrangement.
Define a vector $\textbf{v}\in\{-1,1\}^{2s+1}$ such that for each
$e\in[2s+1]$,
\begin{equation*} v_e=\left\{\begin{aligned}
&~1,~~~~~~~\text{if}~u_i>0~\text{for some}~i\in T_e,\\
&-1,~~~~{\footnotesize~}\text{otherwise}.
\end{aligned}\right.
\end{equation*}
The definition of $\textbf{v}$ is reasonable because by Claim 3,
either $u_i\geq 0$ for all $i\in T_{e}$ or $u_i\leq 0$ for all
$i\in T_{e}$. By \eqref{dff-xxp}, for each $j\in[2s+1]$, we have
\begin{align}\label{diff-ex}
\textbf{x}\cdot\textbf{a}^{(j)}-
\textbf{x}'\cdot\textbf{a}^{(j)}&=\sum_{i=1}^Lu_ii^{j-1}\nonumber\\
&=\sum_{e=1}^{2s+1}\left(\sum_{i\in
T_e}|u_i|i^{j-1}\right)v_e.\end{align}

Since $\textbf{x}'_{[L]\backslash\{i'_{\text{del}}\}}$ can be
obtained from $\textbf{x}_{[L]\backslash\{i_{\text{del}}\}}$ by
substitution of at most $2s$ symbols, then for each $i\in[L]$, it
is not hard to see that
$|u_i|=\left|\sum_{\ell=i}^Lx_i-\sum_{\ell=i}^Lx'_i\right|\leq
2s+1$. Therefore, by \eqref{dff-xxp}, for each $j\in[2s+1]$,
\begin{align*}
\left|\textbf{x}\cdot\textbf{a}^{(j)}-
\textbf{x}'\cdot\textbf{a}^{(j)}\right|&
=\left|\sum_{i=1}^Lu_ii^{j-1}\right|\nonumber\\
&\leq \sum_{i=1}^L(2s+1)i^{j-1}\nonumber\\&\leq
(2s+1)L^j.\end{align*} Now, by \eqref{def-f}, from
$f(\textbf{x})=f(\textbf{x}')$ we can obtain
$\textbf{x}\cdot\textbf{a}^{(j)}-
\textbf{x}'\cdot\textbf{a}^{(j)}=0$ for each $j\in[2s+1]$.
Further, by \eqref{diff-ex}, we have
$$\sum_{e=1}^{2s+2}\left(\sum_{i\in
T_e}|u_i|i^{j-1}\right)v_e=0, ~~\forall j\in[2s+1],$$ or
equivalently,
$$A\textbf{v}^{\intercal}=0,$$ where $A$ is the
$(2s+1)\times(2s+1)$ matrix given by
\begin{align}\label{eq-Mat-A}
A&=\left(\begin{array}{cccc} \sum_{i\in T_1}|u_i|i^0 & \sum_{i\in
T_2}|u_i|i^0 & \cdots &
\sum_{i\in T_{2s+1}}|u_i|i^0 \\
\sum_{i\in T_1}|u_i|i^1 & \sum_{i\in T_2}|u_i|i^1 & \cdots &
\sum_{i\in T_{2s+1}}|u_i|i^1 \\
\vdots & \vdots & \ddots & \vdots \\
\sum_{i\in T_1}|u_i|i^{2s+1} & \sum_{i\in T_2}|u_i|i^{2s+1} &
\cdots & \sum_{i\in T_{2s+1}}|u_i|i^{2s+1} \\
\end{array}\right),
\end{align} and
$\textbf{v}^{\intercal}$ is the transpose of $\textbf{v}$. We will
prove, by similar discussions as in \cite{Sima19-2}, that
$A\textbf{v}^{\intercal}=0$ only when $A=0$, that is, $A$ is the
zero matrix.

Suppose otherwise that $A\neq 0$. Let
$\{e_1,\ldots,e_q\}\subseteq[2s+1]$ be the set of indices of the
columns of $A$ that are non-zero. Consider the submatrix $B$ of
$A$, formed by the intersection of first $Q$ rows of $A$ and
columns $e_1,\ldots,e_q$ of $A$. Then $A\textbf{v}^{\intercal}=0$
implies that $B\textbf{v}_{\{e_1,\ldots,e_q\}}^{\intercal}=0$. By
\eqref{eq-Mat-A}, we have
\begin{align*}
\det(B)&=\det\!\!\left(\!\begin{array}{cccc} \sum_{i\in
T_{e_1}}|u_i|i^{0} & \cdots &
\sum_{i\in T_{e_q}}|u_i|i^{0} \\
\vdots & \ddots & \vdots \\
\sum_{i\in T_{e_1}}|u_i|i^{q-1} &
\cdots & \sum_{i\in T_{e_q}}|u_i|i^{q-1} \\
\end{array}\!\right)\\
&=\sum_{\substack{
i_1\in T_{e_1}, \ldots,\\i_{q}\in T_{e_q}}}
\det\!\!\left(\!\begin{array}{cccc}
|u_{i_1}|i_1^{0} & \cdots & |u_{i_{q}}|i_{q}^{0} \\
\vdots & \ddots & \vdots \\
|u_{i_1}|i_1^{q-1} & \cdots & |u_{i_{q}}|i_{q}^{q-1} \\
\end{array}\!\right) \\
&=\sum_{\substack{
i_1\in T_{e_1}, \ldots,\\i_{q}\in T_{e_q}}}
\!\!\left(\prod_{\ell=1}^{q}|u_{i_\ell}|\right)
\!\det\!\!\left(\!\begin{array}{cccc}
i_1^{0} & \cdots & i_{q}^{0} \\
\vdots & \ddots & \vdots \\
i_1^{q-1} & \cdots & i_{q}^{q-1} \\
\end{array}\!\right) \\
&\stackrel{(\text{i})}{=}\sum_{\substack{
i_1\in T_{e_1}, \ldots,\\i_{q}\in
T_{e_q}}}\!\!\!\left(\prod_{\ell=1}^{q}|u_{i_\ell}|\right)
\!\!\prod_{1\leq e'<e''\leq q}(i_{e''}-i_{e'})\\&>0,
\end{align*} where (i) comes from computing
the determinant of a set of Vandermonde matrices, and the
inequality holds because there is at least one
$(i_1,\ldots,i_{q})$ such that
$\prod_{\ell=1}^{q}|u_{i_\ell}|>0~($Note that $i_1\in T_{e_1},
\ldots,i_{q}\in T_{e_q}$ and $\{e_1,\ldots,e_q\}$ is the set of
indices of the columns of $A$ that are non-zero.$)$ and by
property (P3), $i_{e''}-i_{e'}>0$ for all $i_{e'}\in T_{e'}$ and
$i_{e''}\in T_{e''}$ such that $e'<e''$. Then
$B\textbf{v}_{\{e_1,\ldots,e_q\}}^{\intercal}=0$ implies that
$\textbf{v}_{\{e_1,\ldots,e_q\}}=0$, which contradicts to the
definition of $\textbf{v}$. Thus, it must be the case that $A=0$.
By \eqref{eq-Mat-A}, we have $u_i=0$ for each $e\in[2s+1]$ and
each $i\in T_e$, so $\textbf{x}=\textbf{x}'$, which proves Lemma
\ref{f-protect}.

\subsection{Proof of Lemma \ref{M-cmprsn}}

In this subsection, we prove Lemma \ref{M-cmprsn}. Our proof is
similar to \cite[Lemma 2]{Sima19-2}. We will see that by using the
function $h$ for a pre-coding, the redundancy is made $2s\log
(n_0)$ smaller than the direct construction.

We first need to introduce a new notation. Let $\mathscr C_0$ be
constructed as in Lemma \ref{lem-BCH}. For each
$\textbf{c}\in\mathscr C_0$, let
$$\mathscr N_{\mathscr C_0}(\textbf{c})
\triangleq\big\{\textbf{c}'\in\mathscr C_0: \mathscr
B_{1,s}(\textbf{c})\cap\mathscr
B_{1,s}(\textbf{c}')\neq\emptyset\big\}.$$ Clearly, each
$\textbf{c}'\in\mathscr N_{\mathscr C_0}(\textbf{c})$ can be
obtained through the following four steps: The first step is to
delete one of the $n$ symbols of $\textbf{c}$, i.e., one symbol in
$\{c_1,c_2,\ldots,c_n\}$, to obtain a
$\textbf{y}\in\{0,1\}^{n-1}$, which has $n$ possibilities; the
second step is to insert an $z'\in\{0,1\}$ to $\textbf{y}$ in one
of the $n$ positions in $\textbf{y}$ to obtain a
$\textbf{z}'\in\{0,1\}^{n}$, which has $2n$ possibilities for each
fixed $\textbf{y}$; the third step is to substitute $s'$ elements
$(s'\in[0,s])$ of the $n-1$ elements of $\textbf{z}'~($excluding
the inserted element $z')$ to obtain a $\textbf{z}\in\{0,1\}^n$,
which has $\sum_{s'=0}^{s}\binom{n-1}{s'}$ possibilities for each
fixed $\textbf{z}'$; the fourth step is to substitute $s'$
elements of $\textbf{z}$ to obtain a $\textbf{c}'\in\mathscr
N_{\mathscr C_0}(\textbf{c})$. Note that $s'\leq s$ and by Lemma
\ref{lem-BCH}, $\mathscr C_0$ has minimum $($Hamming$)$ distance
at least $2s+1$. Then for each fixed $\textbf{z}\in\{0,1\}^n$,
there exists at most one $\textbf{c}'\in\mathscr C_0$ such that
$\textbf{c}'$ can be obtained from $\textbf{z}$ by $s'$
substitutions. Hence, we have $($Note that by Lemma \ref{lem-BCH},
$n_0>2s+1.)$
\begin{align}\label{num-C0-Nx}
|\mathscr N_{\mathscr C_0}(\textbf{c})|\leq
2(n_0)^2\sum_{s'=0}^s\binom{n_0-1}{s'}\leq (n_0)^{s+2}.\end{align}

To construct the function $g$, we also need the following claim.

\emph{Claim 4}: Let $f$ be the function constructed by
\eqref{def-f} and $M$ be the mapping constructed by
\eqref{def-M-map}, both with $L=n_0$. There exists a function
$$P:\{0,1\}^{n_0}\rightarrow \big[1,2^{(s+2)
\log (n_0)+o(\log (n_0))}\big],$$ computable in time
$O\left((n_0)^{s+3}\right)$, such that for any $\textbf{c}$,
$\textbf{c}'\in \mathscr C_0$, if $\textbf{c}'\in\mathscr
N_{\mathscr C_0}(\textbf{c})$ and
$\Big(M\big(f(\textbf{c})\big)~\text{mod}~P(\textbf{c}),
P(\textbf{c})\Big)=\Big(M\big(f(\textbf{c}')\big)
~\text{mod}~P(\textbf{c}'), P(\textbf{c}')\Big)$, then
$\textbf{c}=\textbf{c}'$.

\begin{proof}[Proof of Claim 4]
For any $\textbf{c}'\in \mathscr N_{\mathscr
C_0}(\textbf{c})\backslash\{\textbf{c}\}$, by Lemma
\ref{f-protect}, we have $f(\textbf{c})\neq f(\textbf{c}')$, so by
\eqref{def-M-map},
\begin{align}\label{up-bnd-DMF}
0<\left|M\big(f(\textbf{c})\big)-M\big(f(\textbf{c}')\big)\right|
<N,\end{align} where $($noting that $L=n)$
$$N\triangleq(2s+1)^{2s+1}(n_0)^{(s+1)(2s+1)}.$$ Let
\begin{align*}\mathscr P(\textbf{c})=\big\{p{\tiny ~}:~
p ~\text{is a divisor of}~
\left|M\big(f(\textbf{c})\big)-M\big(f(\textbf{c}')\big)\right|
\text{for some}~ \textbf{c}'\in \mathscr N_{\mathscr
C_0}(\textbf{c})\backslash\{\textbf{c}\}\big\}.\end{align*} For
every $\textbf{c}'\in \mathscr N_{\mathscr
C_0}(\textbf{c})\backslash\{\textbf{c}\}$, by \cite[Lemma
7]{Sima19-2}, the number of divisors of
$\left|M\big(f(\textbf{c})\big)-M\big(f(\textbf{c}')\big)\right|$
is upper bounded by
$$2^{1.6\ln N/\ln\ln N}=2^{o(\log (n_0))}.$$
Moreover, by \eqref{num-C0-Nx}, we have $|\mathscr N_{\mathscr
C_0}(\textbf{c})|<(n_0)^{s+2}$. Then we obtain
$$|\mathscr P(\textbf{c})|<(n_0)^{s+2}2^{o(\log (n_0))}=2^{(s+2)\log
(n_0)+o(\log (n_0))},$$ which implies that there exists a number
$P(\textbf{c})\in\left[1,2^{(s+2)\log (n_0)+o(\log (n_0))}\right]$
such that
$$P(\textbf{c})\nmid \left|M\big(f(\textbf{c})\big)
-M\big(f(\textbf{c}')\big)\right|$$ for all $\textbf{c}'\in
\mathscr N_{\mathscr C_0}(\textbf{c})\backslash\{\textbf{c}\}$.
That is,
$$M\big(f(\textbf{c})\big)\not\equiv
M\big(f(\textbf{c}')\big)~\text{mod}~P(\textbf{c})$$ for all
$\textbf{c}'\in \mathscr N_{\mathscr
C_0}(\textbf{c})\backslash\{\textbf{c}\}$. In other words, if
$\textbf{c}'\in \mathscr N_{\mathscr C_0}(\textbf{c})$ and
$M\big(f(\textbf{c})\big)\equiv M\big(f(\textbf{c}')\big)
~\text{mod}~P(\textbf{c})$, then $\textbf{c}=\textbf{c}'$.
Therefore, if $\textbf{c}'\in \mathscr N_{\mathscr
C_0}(\textbf{c})$ and
$\Big(M\big(f(\textbf{c})\big)~\text{mod}~P(\textbf{c}),
P(\textbf{c})\Big)= \Big(M\big(f(\textbf{c}')\big)
~\text{mod}~P(\textbf{c}'), P(\textbf{c}')\Big)$, then we have
$\textbf{c}=\textbf{c}'$. Note that
$P(\textbf{c})\in\left[1,2^{(s+2)\log (n_0)+o(\log
(n_0))}\right]$, so it can be found in time
$O\left((n_0)^{s+3}\right)$ by brute force searching.
\end{proof}

Now, we can construct the function $g$ as follows. Let
$g:\{0,1\}^{n_0}\rightarrow \{0,1\}^{2(s+2)\log (n_0)+o(\log
(n_0))}$ be such that for each $\textbf{c}\in\{0,1\}^{n_0}$,
$g(\textbf{c})=\big(g_1(\textbf{c}), g_2(\textbf{c})\big)$, where
$g_1(\textbf{c})$ is the binary representation of
$M\big(f(\textbf{c})\big)~\text{mod}~P(\textbf{c})$ and
$g_2(\textbf{c})$ is the binary representation of $P(\textbf{c})$.
Since $P(\textbf{c})$ can be computed in time
$O\left((n_0)^{s+3}\right)$, then $g(\textbf{c})$ can be computed
in time $O\left((n_0)^{s+3}\right)$. Moreover, let $h$ be
constructed as in Lemma \ref{lem-BCH}. Then for any $\textbf{x},
\textbf{x}'\in\{0,1\}^k$, by Claim 4, if $h(\textbf{x})\in\mathscr
N_{\mathscr C_0}(h(\textbf{x}'))$ and
$g(h(\textbf{x}))=g(h(\textbf{x}'))$, then
$h(\textbf{x})=h(\textbf{x}')$. Therefore, given
$g(h(\textbf{x}))$ and any $\textbf{y}\in \mathscr
B_{1,s}(\textbf{c})$, $h(\textbf{x})$ is uniquely determined by
$g(h(\textbf{x}))$ and $\textbf{y}$. In other words,
$h(\textbf{x})$ can be recovered from $g(h(\textbf{x}))$ and any
given $\textbf{y}\in \mathscr B_{1,s}(h(\textbf{x}))$. Noticing
that $\left|\big\{\textbf{x}'\in\{0,1\}^k: \textbf{y}\in\mathscr
B_{1,s}(h(\textbf{x}'))\big\}\right|\leq 3(n_0)^{s+1}$ and for
each $\textbf{x}'\in\{0,1\}^k$, the complexity of checking whether
$M\big(f(h(\textbf{x}'))\big)
~\text{mod}~g_2(h(\textbf{x}))=g_1(h(\textbf{x}))$ is $O(n_0)$, so
$h(\textbf{x})$ can be computed from $g(h(\textbf{x}))$ and
$\textbf{y}$ in time $O\left((n_0)^{s+2}\right)$ by brute force
searching. Thus, Lemma \ref{M-cmprsn} is proved.

\section{Conclusions and Discussions}

We proposed a family of systematic single-deletion
$s$-substitution correcting codes of length $n$ with asymptotical
redundancy at most $(3s+4)\log n+o(\log n)$ and encoding/decoding
complexity of $O\left(n^{s+3}\right)$ and $O\left(n^{s+2}\right)$
respectively, where $s\geq 2$ is a constant. The redundancy of our
construction is $s\log n$ less than that the best known
deletion$/$substitution correcting codes.

\subsection{Generalization to More Applications}

The key improvement of our construction is a pre-coding process
using the BCH codes, i.e., the function $h$ constructed by Lemma
\ref{lem-BCH}. This technique can also be generalized to the
construction of $t$-deletion $s$-substitution correcting codes for
$t>1$. In fact, for each information sequence $\textbf{x}$,
although the pre-coding process increases the redundancy by $s\log
n$, it decreases by a factor of $n^s$ the number of vectors
distinct from $\textbf{x}$ which are ``indistinguishable'' with
$\textbf{x}$ in the presence of $t$ deletions and $s$
substitutions. Then a decrease of $2s\log n$ can be achieved by
using the syndrome compression technique \cite{Sima20-1} with
$h(\textbf{x})$. Therefore, the overall redundancy decreases by
$s\log n$ compared to the construction in \cite{Sima20}. Note that
the codes constructed in \cite{Sima20} are capable of correcting
any combination of insertions, deletions and substitutions,
provided that the total number of deletions, insertions, and
substitutions is not greater than $t$, while in this paper,
deletions and substitutions are handled separately, i.e., the
codes are capable of correcting at most $t$ deletions and $s$
substitution, which makes it possible to decrease the number of
``indistinguishable'' sequences by the pre-coding process and
hence decease the redundancy of the resulted codes.

Using a similar approach of pre-coding we can obtain an explicit
construction of $t$-deletion correcting codes whose redundancy is
$(4t-1)\log n$ (improved by $\log n$ compared to the construction
in~\cite{Sima20}). Another possible line of research is nonbinary
$t$-deletion $s$-substitution correcting codes. For nonbinary
case, we can use codes from \cite{Yek2004} for the pre-coding
process.

\appendices


\end{document}